\documentclass[12pt]{amsart}
\usepackage[margin = 1in]{geometry}
\usepackage{subfig, graphicx}
\allowdisplaybreaks
\begin{document}
\title[CAPM with Size and Volatility]{Capital Asset Pricing Model with Size Factor and Normalizing by Volatility Index}
\author{Abraham Atsiwo, Andrey Sarantsev}

\begin{abstract}
The Capital Asset Pricing Model (CAPM) relates a well-diversified stock portfolio to a benchmark portfolio, usually taken to be the S\&P 500. We insert size effect in the CAPM, capturing a real-life feature that on average, small stocks (measured by market capitalization) have higher risk and return than large stocks. Testing CAPM involves fitting linear regressions. Our goal is to ensure that regression residuals are independent identically distributed Gaussian. In some cases, we find that including the Volatility Index as a multiplicative factor by these residuals makes them closer to that ideal. In this article, we combine these ideas to create a new discrete-time model, which includes volatility, size factor, and the CAPM. We fit this model using real-world data, prove the long-term stability of this model, and analyze the resulting capital distribution curves.  
\end{abstract}


\keywords{Capital Asset Pricing Model, stochastic volatility, ergodic Markov process, stationary distribution, size effect, autoregression, capital distribution curve. \textit{Address:} University of Nevada, Reno; Department of Mathematics and Statistics. AA: \texttt{ab30atsiwo@gmail.com}, AS (corresponding author): \texttt{asarantsev@unr.edu}}

\thanks{\textit{Acknowledgements and statements.} We thank the referees for multiple useful comments, which led to significant improvements in this manuscript, and removal of unnecessary parts. The authors received no external funding. The authors have no conflict of interest. All data are taken from publicly available sources and are posted openly in Excel format in GitHub, together with Python code.}

\maketitle

\newtheorem{theorem}{Theorem}
\newtheorem{lemma}{Lemma}
\theoremstyle{definition}
\newtheorem{asmp}{Assumption}
\newtheorem{defn}{Definition}
\thispagestyle{empty}

\section{Introduction} 

\subsection{Ideas in brief}  We remind the readers that {\it price returns} for a stock or a portfolio measure only price changes, ignoring dividends, while {\it total returns} include both price changes and dividends paid. Also, {\it equity premium} is computed as total returns minus risk-free returns (usually measured by short-term Treasury bills). We combine three main ideas in this article. The entire stock market is measured by a certain benchmark, usually taken to be the celebrated Standard \& Poor (S\&P) 500 index. We measure stock size by its market capitalization (number of outstanding stocks multiplied by the price of each stock). We measure portfolio size as the weighted average of stock size, where the weight for each constituent stock matches its weight in this portfolio.

{\it Idea 1: The Capital Asset Pricing Model} (CAPM). We model a target stock portfolio (well-diversified) using a simple linear regression versus the benchmark stock portfolio. This model is de facto a simple linear regression of equity premium of the target portfolio versus the equity premium of the benchmark portfolio. We will also use price returns instead of equity premia in these regressions. The slope and intercept of this regression are called {\it beta} and {\it alpha}, well-accepted terms among financial academics and practitioners. See discussions in a basic textbook \cite{Ang} and a popular book \cite{Siegel}. 

{\it Idea 2: The Size Effect.} Small stocks, on average, have higher volatility but also higher returns than large stocks, see the classic articles \cite{Banz, FF1993}. See also This feature implies long-term stability: Well-diversified stock portfolios stay together. We combined these two ideas in the previous article by the second author \cite{CAPM-size}. This article is a sequel of that work. Such long-term stability is of interest in {\it Stochastic Portfolio Theory} (SPT), which constructs portfolios as functions of market weights.  See the SPT overview in  \cite{FernholzBook, FKSurvey}.

{\it Idea 3: The Volatility Index.} In another article \cite{VIX} by the second author, we observe that total monthly returns of the benchmark (closely tracking the S\&P 500) are not IID (independent identically distributed) Gaussian. However, dividing them by volatility makes them IID Gaussian. We stress this is applied to S\&P 500 not our regression residuals. Here, volatility is measured by the Volatility Index (VIX) monthly average. For the background on the VIX, computed by the Chicago Board of Options Exchange (CBOE), see a textbook \cite{Ang} or a popular book \cite{Siegel}. The volatility itself is modeled as an autoregression of order 1 on the logarithmic scale. Similar observation holds for price monthly returns of large stocks. Finally, all the same is true for the small stocks. We stress that when we say something is true for large or small stocks, we mean a diversified portfolio selected by size, not individual large or small stocks.

\subsection{Our contribution} In this article, we continue research of \cite{CAPM-size} by inserting Idea 3 in this setting, thus combining all three ideas. We create a new discrete-time model, and fit it using real market data from Kenneth French's data library and Federal Reserve Economic Data (FRED) web site. We then state and prove long-term stability of this market model, and interpret this in terms of the SPT. Unlike in our prevous article \cite{CAPM-size}, we consider only discrete-time models in this article. We fill some gaps left in the previous research from \cite{CAPM-size}.

\subsection{Organization of the article} In \textsc{Section}~2, we provide a comprehensive motivation of the proposed model, detailed description of all its components, and a literature review. \textsc{Section}~3 is devoted to real-world financial data description, and the statistical fit of these models. In \textsc{Section}~4, we state and prove the main long-term stability (ergodicity) result for this model consisting of one target portfolio and the benchmark portfolio. This is  \textsc{Theorem}~\ref{thm:main}. This section does not contain any simulations or empirical data analysis. In \textsc{Section}~5, we move to the SPT, considering a similar model but with one benchmark and many target portfolios. We simulate the capital distribution curve (market weights ranked by size plotted versus their ranks). We state and prove a long-term stability result for this curve in \textsc{Theorem}~\ref{thm:weights}. Next, \textsc{Theorem}~\ref{thm:normal-curve} contains results about its shape: We study it by reducing it to normal order statistics. Finally, \textsc{Section}~6 contains final conclusions and outlines for possible further research. All data and code are available on \texttt{GitHub} repository \texttt{asarantsev/size-capm-vix} and all data are open-source. 

\section{Background and Motivation}

\subsection{Capital Asset Pricing Model} This celebrated model, abbreviated as CAPM, compares returns of a stock portfolio with returns of the benchmark. This model was proposed by \cite{Lintner, Mossin, Sharpe}. The benchmark is usually taken to be the Standard \& Poor 500 for the American stock market. The model states that the only factor which matters for a well-diversified portfolio is {\it market exposure}, otherwise known by a standard term {\it beta} and denoted by $\beta$. The case $\beta = 0$ corresponds to the risk-free portfolio, with guaranteed deterministic return. This is usually measured using a benchmark of a short-term rate $r$, for example 1-month Treasury rate. The case $\beta = 1$ corresponds to the market portfolio (the benchmark). When $\beta \in (0, 1)$, the stock portfolio can be replicated by investing in a portfolio of risk-free bonds and the stock market benchmark, in proportions $1 - \beta$ and $\beta$, respectively. In other words, total returns $Q$ (including price changes and dividends) of this portfolio are related to total returns $Q_0$ of the benchmark:
\begin{equation}
\label{eq:CAPM-Q}
Q = (1-\beta)r + \beta Q_0.
\end{equation}
Equivalently, we can rewrite~\eqref{eq:CAPM-Q} in terms of {\it equity premia} $P = Q - r$ of the portfolio and $P_0 = Q_0 - r$ of the benchmark: 
\begin{equation}
\label{eq:CAPM}
P = \beta P_0.
\end{equation}
If $\beta > 1$, this~\eqref{eq:CAPM} still holds, and can be interpreted as shorting bonds and investing everything in the benchmark. We can treat $\beta$ as a risk measure: $\beta > 1$ means that the portfolio is riskier than the benchmark, and $\beta \in (0, 1)$ means the opposite. The case $\beta < 0$ does not usually happen in practice, see \cite[Chapter 7]{Ang}.

It is not considered a big achievement if a money manager improved returns by increasing $\beta$. In fact, often these managers are expected to maximize {\it excess return:} Total returns of a portfolio adjusted for market exposure. This quantity is denoted by $\alpha$ and, accordingly, is called {\it alpha}. These two Greek letters $\alpha$ and $\beta$ are standard notation in Finance. This methodology of market exposure and excess return is well-accepted by both finance academics and practitioners. One can include $\alpha$ into~\eqref{eq:CAPM} as
\begin{equation}
\label{eq:alpha-beta}
P = \alpha + \beta P_0 + \varepsilon,\quad \mathbb E\varepsilon = 0.
\end{equation}
We also include an error term $\varepsilon$, since the model might not hold almost surely. This makes~\eqref{eq:alpha-beta} a simple linear regression of $P$ upon $P_0$. 

Subsequent research disproved  the strong claims of CAPM that market exposure $\beta$ is the only risk measure and quantity of interest for a diversified stock portfolio, see \cite{FF1993, FF2004}, and \cite[Chapter 7]{Ang}. For one, there are systematic ways to generate excess return $\alpha$ by using several factors. Also, $\beta$ might be unstable in the long run. Still, $\beta$ is an established risk measure, accepted by financial theorists, analysts, and managers. The CAPM is useful as a benchmark model, a starting point for more complicated and real-life models. Calculating $\beta$ for mutual funds and exchange-traded funds is common practice. 

\subsection{Size and value} The most well-known and accepted factors are {\it size} (average market capitalization of portfolio stocks) and {\it value} (fundamentals such as earnings, dividends, or book value, compared to price). These are well-accepted by financial academic community and are considered useful by industry practitioners, to the extent that there are multiple size- and value-based funds traded alongside the S\&P 500 funds. See the classic articles \cite{Banz, FF1993} and a more recent article \cite{FF2015}. Also, see the discussion in \cite{Ang} and \cite{Siegel}. 

Including a few factors (for example size and value) would enrich~\eqref{eq:alpha-beta}. In particular, size factor is related to $\beta$ as follows: Well-diversified portfolios of small stocks have equity premia $P$ closely correlated with that $P_0$ of large stock benchmark S\&P 500. This simple linear regression of $P$ vs $P_0$ has very large $R^2$ but $\beta$ which is slightly larger than $1$. This can be done, for example, with exchange-traded funds tracking small, mid, and large (=benchmark) stock indices, see \cite[Appendix]{CAPM-size}. The $\beta$ for a mid-cap/small-cap index is 1.15/1.27.  

One could model the $\beta$ as a function of a portfolio size relative to benchmark. Additionally, we could model the $\alpha$ similarly, as a function of a portfolio size relative to benchmark. However, this might not be as important as for the $\beta$. Indeed, in \cite[Appendix]{CAPM-size} we have $\alpha$ not significantly different from zero for both small-cap and mid-cap indices. See more on the size effect in \cite{Semenov}, and the discussion in \cite[Chapter 8]{Ang}.

\subsection{CAPM-based model with size as factor} Therefore, we developed the following model in \cite{CAPM-size}: Let $S$ = market capitalization of target portfolio, and $S_0$ = market capitalization of benchmark portfolio. Then the {\it relative size} (on the log scale) is defined as 
\begin{equation}
\label{eq:relsize}
C = \ln(S/S_0)
\end{equation}
For $C = 0$, the relative size is 0 (on the log scale) or 1 (on the absolute scale). This corresponds to the target portfolio having the same properties as the benchmark portfolio, with $\alpha = 0$ and $\beta = 1$. The simplest model is linear: For some coefficients $a, b$,
\begin{equation}
\label{eq:alpha-beta-functions}
\alpha = aC, \quad \beta = bC + 1
\end{equation}
In fact, in \cite{CAPM-size} we have more general conditions: $\alpha$ and $\beta$ are general functions of $C$. Here, we focus only on linear functions, see \cite[Example 1]{CAPM-size}. Analysis for small-cap and mid-cap funds above shows that $a \approx 0$ but $b > 0$. Rewrite~\eqref{eq:alpha-beta} as follows by plugging there~\eqref{eq:alpha-beta-functions}:
\begin{equation}
\label{eq:CAPM-size}
P = AC + \left(1 + BC\right)P_0 + \delta,
\end{equation}
where $\delta$ are IID regression residuals with mean zero. A good way to generalize~\eqref{eq:CAPM-size} is to make it dynamic: a time series. To this end, we make the model complete by writing an equation for $S, S_0, P_0$. Unlike total returns, which include dividends, price returns are computed only using price movements. The main idea is to take CAPM linear regression, and replace equity premia with price returns. We get
\begin{equation}
\label{eq:returns-size}
R = aC + \left(1 + bC\right)R_0 + \varepsilon,
\end{equation}
where $a, b$ are some  coefficients (not necessarily the same as the $A, B$ from~\eqref{eq:CAPM-size}), and $\varepsilon$ are IID regression residuals with mean zero (not necessarily the same as $\varepsilon$, but possibly correlated with these). We can interpret change in logarithm as price returns: 
\begin{align}
\label{eq:d-ln}
\begin{split}
R(t) &= \ln S(t+1) - \ln S(t) = \ln\frac{S(t+1)}{S(t)},\\ 
R_0(t) &= \ln S_0(t+1) - \ln S_0(t) = \ln\frac{S_0(t+1)}{S_0(t)}.
\end{split}
\end{align}
In real-life finance, small stocks grow, on average, faster than large stocks. Thus formerly small stocks might become large, and formerly large stocks might become small. The relative size of a stock portfolio exhibits mean-reversion. In our article \cite{CAPM-size}, we proved this for continuous time and more general systems than~\eqref{eq:CAPM-size},~\eqref{eq:returns-size}, under the assumptions that the benchmark follows a lognormal Samuelson (Black-Scholes) model of geometric random walk with growth rate $g$ and total returns $G$, and 
\begin{equation}
\label{eq:stability}
a + bg < 0.
\end{equation}
This is consistent with the observation made above that $a \approx 0$ and $b < 0$, since $g > 0$. Analysis of real-world finance data in \cite{CAPM-size} showed that $A \approx a$ and $B \approx b$. 

\subsection{Gaussian white noise} This is the classic assumption for (simple and multiple) linear regression fitted by the {\it ordinary least squares} (OLS) method: Residuals (error terms) must be IID Gaussian. Indeed, dependence makes the OLS method no longer the best, instead suggesting a switch to generalized least squares (GLS); if residuals are not identically distributed, this makes them {\it heteroscedastic} (having different variance), again contradicting the OLS method; and Gaussian assumption provides the basis of $T$-testing and confidence intervals for coefficient estimates. 

More on normality testing and ACF is in \cite{FanYao}. Below, we summarize  known methods.

Testing for Gaussianity is done via statistical tests, for example Shapiro-Wilk or Jarque-Bera. Alternatively, we can use graphical means: Comparing the histogram of data with the bell curve; or using the quantile-quantile plot of the data versus the corresponding normal distribution, with the same mean and standard deviation as the data.

Testing for IID is harder. One can apply the common white noise tests based on the {\it autocorrelation function} (ACF), which is computed as the correlation of the data with itself after being shifted by a certain lag. A weighted sum of squares of such autocorrelations is then compared with the 
$\chi^2$ distribution. We remark that it is not enough to test the original data using the ACF: The data might have zero ACF; but the squares or the absolute values of this data might have significant positive ACF. In fact, just such an example is shown below in \textsc{Figure}~\ref{fig:premia}.  Alternatively, one can simply plot the empirical ACF for the data. In common statistical packages, the values of the ACF are plotted against a strip where we fail to reject the white noise null hypothesis. If this value of the ACF for this lag is outside of this strip, we reject the null hypothesis for this fixed lag. The Ljung-Box test plays the role of an omnibus test which combines various lags. Of course, one should also plot the ACF for the squares or the absolute values of the data.

\subsection{The Volatility Index} However, there is a drawback in our modeling from \cite{CAPM-size}. We fit regressions~\eqref{eq:CAPM-size} and~\eqref{eq:returns-size} using real-world monthly data taken from Kenneth French's Dartmouth College Financial Data Library. Residuals $\varepsilon, \delta$ are not IID. Absolute values are autocorrelated. This feature is also true for S\&P 500 monthly returns themselves, see our recent article \cite{VIX}. Also, these returns are not Gaussian. We wish to improve the fit of regressions~\eqref{eq:CAPM-size} and~\eqref{eq:returns-size} to make residuals closer to IID Gaussian. 

For the S\&P 500 returns, we did this in \cite{VIX} as follows. We divided these returns by monthly average VIX: The S\&P 500 Volatility Index $V$ computed daily by the CBOE. The computation is done as a weighted average of the {\it implied volatility} for each call or put option traded on the CBOE. In turn, this implied volatility is the solution of the Black-Scholes equation given the market price of this option and unknown volatility. Our main idea of this article \cite{VIX} is (in our notation, see \cite[Subsection 2.2]{VIX}):
\begin{equation}
\label{eq:norm-return}
\frac{R_0(t)}{V(t)} \sim \mbox{IID Gaussian}.
\end{equation}
Similarly, it is reasonable to model normalized equity premia as IID Gaussian:
\begin{equation}
\label{eq:norm-premia}
\frac{P_0(t)}{V(t)} \sim \mbox{IID Gaussian}.
\end{equation}
We did not do this in \cite{VIX}, but we complete this work in this article. For the original equity premia of benchmark $P_0$, we plot in \textsc{Figure}~\ref{fig:premia} the quantile-quantile plot and the autocorrelation function for $P_0$ and $|P_0|$. For the normalized equity premia $P_0/V$, we plot these in \textsc{Figure}~\ref{fig:npremia}. We see that division by VIX is needed to model equity premia as in~\eqref{eq:norm-premia}. Data for $P_0$ is taken from French's data library: Top 10\% decile, January 1986--December 2025 (exactly 40 years and 480 months). For short-term rate $r$, we use the 3-month Treasury rate from FRED,  start of month data. As discussed earlier, the data and Python code is available on \texttt{Github} repository: \texttt{asarantsev/size-capm-vix} see files \texttt{data.xlsx} and \texttt{equity-premia.py} We can see from \textsc{Figure}~\ref{fig:npremia} that it is reasonable to model $P_0/V$ as IID Gaussian. The same is true for $R_0/V$. 

\begin{figure}[t]
\subfloat[QQ]{\includegraphics[width = 5cm]{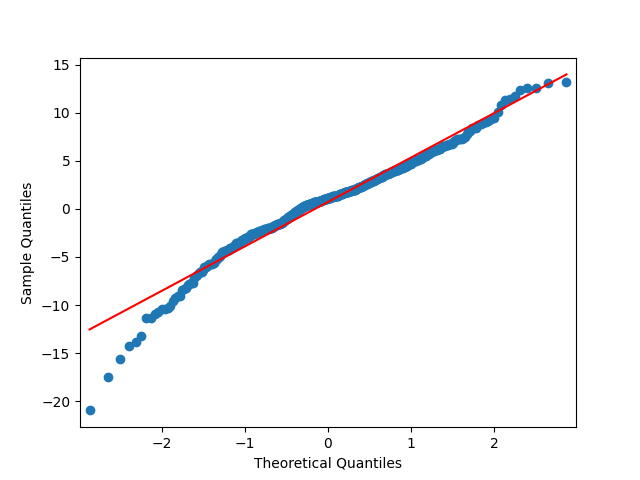}}
\subfloat[ACF]{\includegraphics[width = 5cm]{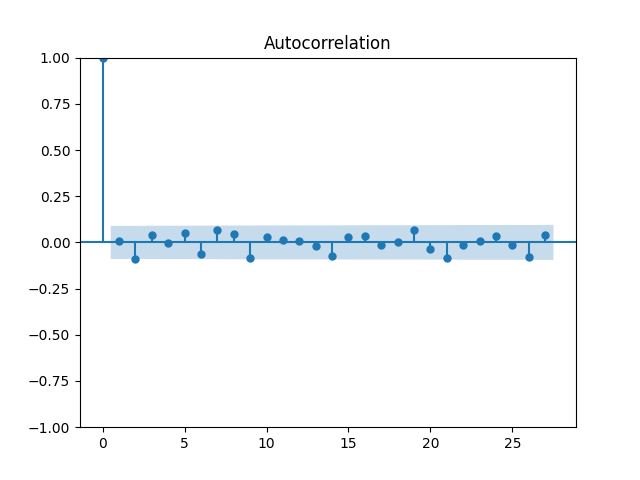}}
\subfloat[ACF]{\includegraphics[width = 5cm]{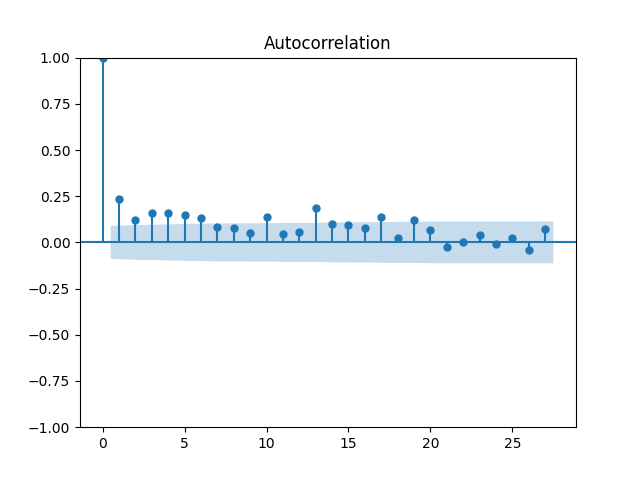}}
\caption{Empirical real-world data: The quantile-quantile (QQ) plot for equity premia $P_0$, and the empirical autocorrelation function (ACF) for $P_0$ and for $|P_0|$, given the monthly data 1986--2025 for the benchmark: equal-weighted portfolio of the top decile $P_0$. We see it is not enough to test the potential IID by plotting ACF: Even if it is close to zero, the ACF for the absolute values might still be significant. }
\label{fig:premia}
\end{figure}

\begin{figure}[t]
\subfloat[QQ]{\includegraphics[width = 5cm]{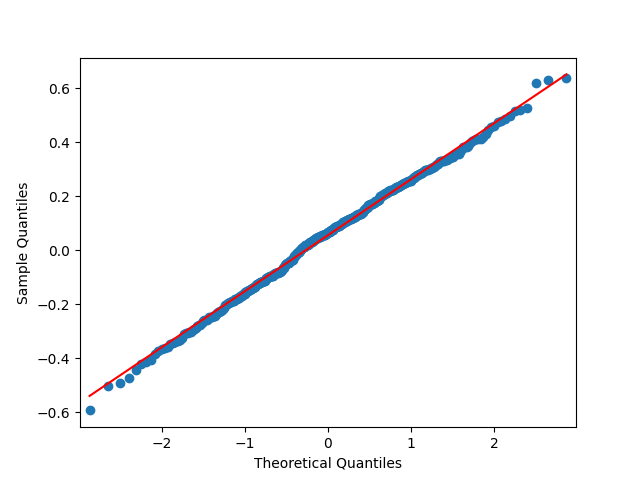}}
\subfloat[ACF]{\includegraphics[width = 5cm]{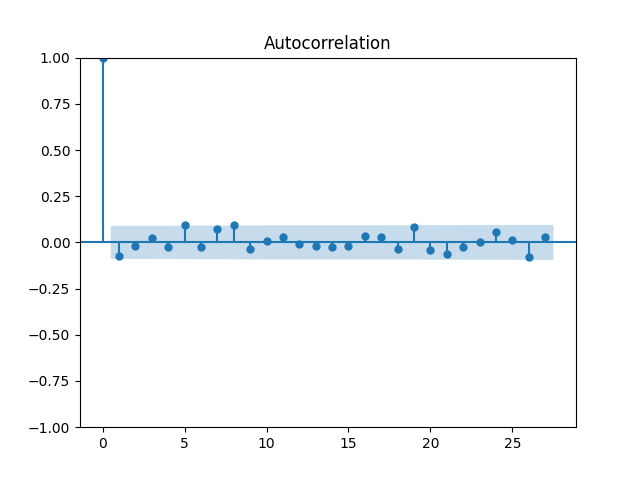}}
\subfloat[ACF]{\includegraphics[width = 5cm]{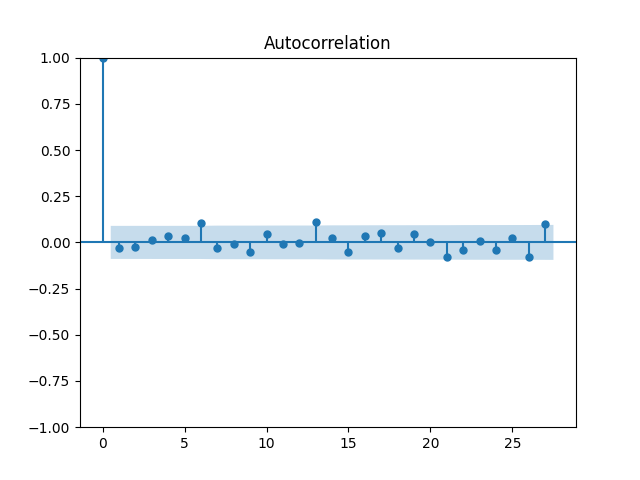}}
\caption{Empirical real-world data: The quantile-quantile (QQ) plot for equity premia $P_0/V$, and the empirical autocorrelation function (ACF) for $P_0/V$ and for $|P_0/V|$, given the monthly data 1986--2025 equal-weighted portfolio of the benchmark: the top decile $P_0$, normalized by monthly VIX $V$.}
\label{fig:npremia}
\end{figure}

The VIX itself is modeled by an autoregression of order 1 on the log scale, see \cite{VIX}:
\begin{equation}
\label{eq:ln-VIX}
\ln V(t) = \alpha + \beta \ln V(t-1) + W(t).
\end{equation}
Slightly abusing the notation, but following \cite{VIX}, in the rest of the article we use $\alpha$ and $\beta$ as the intercept and the slope of this autoregression~\eqref{eq:ln-VIX}, instead of excess return and market exposure from the CAPM. This model~\eqref{eq:ln-VIX} has good fit, as shown in \cite[Section 3]{VIX}. Innovations $W$ are IID but not Gaussian, with some finite exponential moments. The point estimate $\hat{\beta} = 0.88$, and we reject the unit root hypothesis $\beta = 1$. The methodology of \cite{VIX} stands in contrast to classic stochastic volatility models, see for example \cite{SV-book}, where the volatility $V$ is not observed directly and must be inferred from the data $R_0$ or $P_0$. 

\subsection{Capital distribution curve} A newly developed framework for portfolio management is the SPT, see the classic book \cite{FernholzBook}. One question which concerns it is {\it model stability}: In a model of $N$ stocks, do they move together in the long run, or do they split into several subsets, moving away from each other in the long run? A related question is analysis of {\it market weights:} A market weight $\mu$ of a stock is its market capitalization (size) divided by the sum of all market capitalizations. Rank market weights at each time from top to bottom: 
$$
\mu_{(1)}(t) \ge \ldots \ge \mu_{(N)}(t).
$$
The plot of $N$ points
$$
(\ln n, \ln\mu_{(n)}(t)),\, n = 1, \ldots, N
$$
is called the {\it capital distribution curve}. For real-world markets, this curve is concave and straight at left upper end. See the picture in \cite[Chapter 4]{FernholzBook} for 1929, 1939, \ldots, 1999. If the model is stable, the capital distribution curve is also stable in the long run: see \textsc{Theorem}~\ref{thm:weights}. Previous research analyzed long-term stability and capital distribution curve for various continuous-time models: \cite{Chatterjee, Pal}. This is still an active area of research, for example the recent manuscript \cite{Baker} discusses time stability of the capital distribution curve, making interesting connections with physics. These models were designed to capture the observations that, on average, small stocks have higher volatility and growth rates than large stocks. In \cite{Chatterjee}, competing Brownian particle models were analyzed, where drift and diffusion coefficients depend on the rank of the stock, and linearity of the capital distribution curve was reproduced. In \cite{Pal}, the drift and diffusion depend on the market weight of a stock, but the capital distribution curve is not linear. These models do not use CAPM or VIX. 

In the previous article by the second coauthor \cite{CAPM-size}, devoted to the CAPM with size but without VIX (in other words, with constant volatility), we also reproduced model stability for its continuous-time version \cite[Theorem 2]{CAPM-size}. Using simulation, we reproduced the linearity of the capital distribution curve. But we did not state and prove rigorous results on this. A natural question is to reproduce these results for our model here. We accomplish both tasks in this article: (a) we simulate the capital distribution curve in Subsection 5.2, which reproduces its linearity; (b) we prove results on convergence to a Poisson point process. This fills a gap in \cite{CAPM-size}. 

\subsection{Our contributions} In this article, we combine the ideas of our articles \cite{CAPM-size} on CAPM and \cite{VIX} on VIX to state a reasonable generalization of \cite{CAPM-size}. This model can be {\it truncated}: Include only price changes for the benchmark $R_0$ and the target $R$ (or, equivalently, market size $S_0$ and $S$), and volatility $V$. Alternatively, it can be {\it completed}: Include also equity premia $P_0, P$ for the benchmark and the target. A truncated model contains 3 time series, and a completed model contains 5 time series. 

In \cite{CAPM-size}, we model~\eqref{eq:returns-size} and~\eqref{eq:CAPM-size} but with $P_0$ or $S_0$ IID Gaussian, in continuous time. The article \cite{CAPM-size} does not include VIX $V$. However, this article is concerned with the idea from \cite{VIX} that division by VIX normalized returns and premia. We replace $P$ and $P_0$ with $P/V$ and $P_0/V$ in~\eqref{eq:CAPM-size}. Also, we replace $R$ with $R/V$ and $R_0$ with $R_0/V$ in~\eqref{eq:returns-size}. We model $(R_0/V, P_0/V)$ as IID (but maybe not Gaussian), following~\eqref{eq:norm-return} and~\eqref{eq:norm-premia}. As mentioned in the Introduction, in \textsc{Section}~3 we perform the main fitting of linear regressions. 

We fill two lacunas left in our research \cite{CAPM-size}. The first lacuna is stability results for discrete time. In \textsc{Section}~4 of this article, we state and prove this for the case of stochastic volatility. The second lacuna is rigorous results for the capital distribution curve, which are lacking in \cite{CAPM-size}. In \textsc{Section}~5 of this article, we state and prove a convergence result for the capital distribution curve: \textsc{Theorem}~\ref{thm:weights}. Next, in \textsc{Theorem}~\ref{thm:normal-curve}, we show a remarkable result: the curve is 
$$
(\ln n, X_{(n)}),\quad \mbox{for}\quad n = 1, \ldots, N,
$$
where $X_{(1)} \ge \ldots \ge X_{(N)}$ are order statistics of a conditional normal sample: 
$$
X_1, \ldots, X_N \mid \mu, \sigma \sim \mathcal N(\mu, \sigma^2)
$$
for random $\mu$ and $\sigma$. We further analyze this curve for the standard normal $X_1, \ldots, X_N$, and show its upper left and lower right ends replicate real-world behavior. 

\section{Financial Data and Statistical Analysis}

\subsection{Data description} We take monthly data January 1986--December 2025. This gives us exactly 40 years and 480 data points. As discussed in the Introduction and Background sections, we measure volatility $V$ as CBOE VIX: The Volatility Index for S\&P 500, monthly average data. This data is taken from Federal Reserve Economic Data (FRED) web site. For equity premia computation, we need short-term Treasury rates, which are also taken from the FRED web site. 

The rest of the data is taken from French's data library. It contains equally-weighted portfolios of stocks split into 10 deciles by size: Decile 1 has smallest 10\% stocks by size (market capitalization), Decile 2 has the next 10\% smallest stocks, etc up to Decile 10, which contains top 10\% largest stocks. We use Decile 10 as the benchmark. We stress that this equally-weighted portfolio of Decile 10 is different from the classic benchmark S\&P 500. These portfolios are reconstituted at the end of June of each year. For each decile and each month, the data contains average market size, price returns (excluding dividends), and total returns (including dividends). The data are available on the \texttt{GitHub} repository \texttt{asarantsev/size-capm-vix} in the file \texttt{data.xlsx}

\subsection{Regression description} Our goal is to fit~\eqref{eq:returns-size} and~\eqref{eq:CAPM-size}. We rewrite these as
$$
\frac{R(t) - R_0(t)}{V(t)} = aC(t) + bC(t)\frac{R_0(t)}{V(t)} + \varepsilon(t).
$$
$$
\frac{P(t) - P_0(t)}{V(t)} = AC(t) + BC(t)\frac{P_0(t)}{V(t)} + \delta(t).
$$
These linear regressions do not have an intercept, however. To make the models complete, we add intercepts:
\begin{equation}
\label{eq:returns-reg}
\frac{R(t) - R_0(t)}{V(t)} = aC(t) + bC(t)\frac{R_0(t)}{V(t)} + m + \varepsilon(t).
\end{equation}
\begin{equation}
\label{eq:premia-reg}
\frac{P(t) - P_0(t)}{V(t)} = AC(t) + BC(t)\frac{P_0(t)}{V(t)} + M + \delta(t).
\end{equation}
For the benchmark with price returns $R_0$ and equity premia $P_0$, we take Decile 10. For the target with price returns $R$ and equity premia $P$, we use Deciles 1, \ldots, 9. Each decile is fit separately. See results in \textsc{Table}~\ref{table:return} and  \textsc{Table}~\ref{table:premia}. The Python code is available on the \texttt{GitHub} repository \texttt{asarantsev/size-capm-vix} in the file \texttt{regression-fit.py}

\begin{table}
\begin{tabular}{|c|c|c|c|c|c|c|c|c|c|c|}
\hline
Decile & $m$ & $T_{m}$ & $a$ & $T_a$ & $b$ & $T_b$ & $s$ & $R^2$ & LB & SW\\
\hline
1 & -.12 & -0.64 & .021 & 0.75 & -.025 & -3 & .25 & 2\% & 0.5\% & 0.0\%\\
\hline
2 & -.37 & -2.57 & .072 & 2.57 & .012 & 1.3 & .22 & 1.7\% & 32\% & 0.0\%\\
\hline
3 & -.21 & -1.64 & .047 & 1.68 & .020 & 2.16 & .19 & 1.5\% & 64\% & 1.3\%\\
\hline
4 & -.25 & -2.08 & .06 & 2.09 & .02 & 2.17 & .18 & 1.8\% & 27\% & 0.2\%\\
\hline
5 & -.17 & -1.77 & .046 & 1.78 & .033 & 3.56 & .16 & 3.1\% & 54\% & 0.7\%\\
\hline
6 & -.027 & -.41 & .009 & .46 & .025 & 2.92 & .13 & 1.8\% & 63\% & 46\%\\
\hline
7 & -.086 & -1.78 & .032 & 1.88 &  .027 & 3.2 & .11 & 2.8\% & 90\% & 16\%\\
\hline
8 & -.055 & -1.48 & .025 & 1.54 & .03 & 3.46 & .09 & 2.9\% & 15\% & 14\%\\
\hline
9 & -.02 & -.82 & .015 & .93 & .01 & 1.11 & .07 & 0.4\% & 22\% & 1.6\%\\
\hline
\end{tabular}
\bigskip
\caption{Results of ordinary least squares fit for the regression~\eqref{eq:returns-reg}. Point estimates of parameters $m, a, b$, the corresponding Student $T$-values, the standard error $s$ of residuals $\varepsilon$, the $R^2$ of regression, and the $p$-values for two tests for residuals: the Ljung-Box (LB) white noise test with 10 lags, and the Shapiro-Wilk (SW) normality test.}
\label{table:return}
\end{table}

\begin{table}
\begin{tabular}{|c|c|c|c|c|c|c|c|c|c|c|}
\hline
Decile & $M$ & $T_M$ & $A$ & $T_A$ & $B$ & $T_B$ & $s$  & $R^2$ & LB & SW\\
\hline
1 & -.15 & -.78 & .024 & .85 & -.025 & -3.1 & .25 & 2\% & 0.5\% & 0.0\%\\
\hline
2 & -.4 & -2.8 & .077 & 2.75 & .012 & 1.3 & .22 & 1.9\% & 30\% & 0.0\%\\
\hline
3 & -.24 & -1.86 & .053 & 1.86 & .020 & 2.18 & .19 & 1.7\% & 62\% & 1.6\%\\
\hline
4 & -.27 & -2.3 & .066 & 2.28 & .02 & 2.2 & .17 & 2.0\% & 24\% & 0.2\% \\
\hline
5 & -.19 & -2 & .051 & 2 & .033 & 3.59 & .16 & 3.3\% & 51\% & 0.7\%\\
\hline
6 & -.044 & -.68 & .014 & .69 & .026 & 2.94 & .13 & 1.9\% & 64\% & 45\%\\
\hline
7 & -.093 & -1.9 & .034 & 1.96 & .026 & 3.13 & .11 & 2.8\% & 86\% & 12\%\\
\hline
8 & -0.057 & -1.55 & .025 & 1.56 & .03 & 3.53 & .09 & 3\% & 19\% & 11\%\\
\hline
9 & -.02 & -.79 & .014 & .87 & .01 & 1.12 & .07& 0.4\% & 17\% & 1\% \\
\hline
\end{tabular}
\bigskip
\caption{Results of ordinary least squares fit for the regression~\eqref{eq:premia-reg}. Point estimates of parameters $M, A, B$, the corresponding Student $T$-values, the standard error $s$ of residuals $\delta$, the $R^2$ of regression, and the $p$-values for two tests for residuals: the Ljung-Box (LB) white noise test with 10 lags, and the Shapiro-Wilk (SW) normality test.}
\label{table:premia}
\end{table}

\subsection{Regression results} Choosing the standard 5\% threshold for statistical significance (and ignoring for the moment the problem of multiple simultaneous testing), we see that residuals can be modeled as IID Gaussian for deciles 6--8. Residuals pass the white noise test for all but the first decile. This holds for both price returns and equity premia. 

\section{Long-Term Stability} 

\subsection{Formal construction of the model} Take a sequence of five-dimensional vectors:
\begin{equation}
\label{eq:innovations}
\mathbf{Y}(t) := (W(t), Z(t), U(t), \delta(t), \varepsilon(t)),\quad t = 1, 2, \ldots
\end{equation}

\begin{asmp}
Vectors $\mathbf{Y}(t)$ are IID with mean zero and finite second moment, with Lebesgue joint density on $\mathbb R^5$ which is everywhere strictly positive. 
\label{asmp:innovations}
\end{asmp}

These five components might be correlated between themselves. First, we model the volatility $V$ using~\eqref{eq:ln-VIX} with innovations $W$. Next, we model $R_0$ following~\eqref{eq:norm-return} for some constant $g \in \mathbb R$: 
\begin{equation}
\label{eq:n-return}
\frac{R_0(t)}{V(t)} = g + U(t),
\end{equation}
but we do not necessarily assume $U$ is Gaussian. We add $g$ because Assumption~\ref{asmp:innovations} states that $\mathbb E[U(t)] = 0$, but the left-hand side of~\eqref{eq:n-return} has nonzero mean. Similarly, we model $P_0$ following~\eqref{eq:norm-premia} for another constant $G \in \mathbb R$:
\begin{equation}
\label{eq:n-premia}
\frac{P_0(t)}{V(t)} = G + Z(t),
\end{equation}
and $Z$ might not be Gaussian. These two equations~\eqref{eq:n-return} and~\eqref{eq:n-premia} model the benchmark. Next, we combine~\eqref{eq:CAPM-size} and~\eqref{eq:returns-size} with these new ideas, following the outline in \textsc{Section}~2. We can replace $P_0$ and $P$ with $P_0/V$ and $P/V$ in~\eqref{eq:CAPM-size}:
\begin{equation}
\label{eq:premia-main}
\frac{P(t)}{V(t)} = M + AC(t) + (1 + BC(t))\frac{P_0(t)}{V(t)} + \delta(t).
\end{equation}
Similarly, we can replace $R_0$ and $R$ with $R_0/V$ and $R/V$ in~\eqref{eq:returns-size}:
\begin{equation}
\label{eq:return-main}
\frac{R(t)}{V(t)} = m + aC(t) + (1 + bC(t))\frac{R_0(t)}{V(t)} + \varepsilon(t).
\end{equation}
Recall the definition of the {\it relative size} process 
\begin{equation}
\label{eq:rel-size}
C(t) = \ln\frac{S(t)}{S_0(t)}.
\end{equation} 
The main model thus becomes~\eqref{eq:ln-VIX},~\eqref{eq:n-return},~\eqref{eq:n-premia},~\eqref{eq:premia-main},~\eqref{eq:return-main},~\eqref{eq:rel-size}.

\subsection{Stability results} Consider a discrete-time process $\mathbf{X} = (\mathbf{X}(0), \mathbf{X}(1), \ldots)$ in $\mathbb R^d$.

\begin{defn}
This process $\mathbf{X}$ is called {\it time-homogeneous Markov} if there exists a {\it transition function} $\mathcal Q : \mathbb R^d \times \mathcal B \to [0, 1]$ (where $\mathcal B$ is the Borel $\sigma$-algebra on $\mathbb R^d$) such that for all $t = 1, 2, \ldots$, $\bf{x}_0, \bf{x}_1, \ldots, \bf{x}_{t-1} \in \mathbb R^d$, and $A \in \mathcal B$, we have:
$$
\mathbb P(\mathbf{X}(t) \in A\mid \mathbf{X}(0) = \bf{x}_0, \mathbf{X}(1) = \bf{x}_1, \ldots, \mathbf{X}(t-1) = \bf{x}_{t-1}) = \mathcal Q(\bf{x}_{t-1}, A).
$$
\end{defn}

\begin{lemma} Under Assumption~\ref{asmp:innovations}, the process $C$ is Markov. Also, the truncated model $(V, R_0, R)$ is Markov. Finally, the completed model $(V, R_0, R, P_0, P)$ is Markov. 
\end{lemma}

\begin{proof} It is clear from~\eqref{eq:ln-VIX} that the process $\ln V$ (and therefore $V$) is Markov. Together with~\eqref{eq:n-return}, this shows that $(V, R_0)$ is Markov. From definition~\eqref{eq:rel-size}, we write 
\begin{align}
\label{eq:dc}
\begin{split}
C(t+1) - C(t) &= \ln\frac{S(t+1)}{S_0(t+1)} - \ln\frac{S(t)}{S_0(t)} \\ & = \ln\frac{S(t+1)}{S(t)} - \ln\frac{S_0(t+1)}{S_0(t)} = R(t) - R_0(t).
\end{split}
\end{align}
Using~\eqref{eq:dc}, we rearrange~\eqref{eq:return-main} as follows:
\begin{equation}
\label{eq:C-main}
C(t+1) = (1 + aV(t) + bR_0(t))C(t) + V(t)(m + \varepsilon(t)).
\end{equation}
But from Assumption~\ref{asmp:innovations}, this process $C$ from~\eqref{eq:C-main} is Markov, too. From~\eqref{eq:dc} and~\eqref{eq:n-return}, this equation~\eqref{eq:C-main} shows that $(V, R_0, C)$, or, equivalently, $(V, R_0, R)$, is Markov. Finally, from~\eqref{eq:premia-main} and~\eqref{eq:n-premia}, we get that the completed model is Markov. 
\end{proof}

\begin{defn} A time-homogeneous Markov process $\mathbf{X}$ has a {\it stationary distribution} $\pi$ if $\pi$ is a probability measure on 
$\mathbb R^d$, and from $\mathbf{X}(0) \sim \pi$ it follows that $\mathbf{X}(1) \sim \pi$ (and therefore $\mathbf{X}(t) \sim \pi$ for all $t$). Equivalently, in terms of transition function $\mathcal Q$: For every $A \in \mathcal B$, 
$$
\int_{\mathbb R^d}\mathcal Q(x, A)\pi(\mathrm{d}x) = \pi(A).
$$
\end{defn}

\begin{defn} A time-homogeneous Markov process $\mathbf{X}$ is {\it ergodic} if it has a unique stationary distribution $\pi$, and for every $x \in \mathbb R^d$, we have:
$$
\sup\limits_{A \subseteq \mathbb R^d}\Bigl|\mathbb P(\mathbf{X}(t) \in A\mid \mathbf{X}(0) = x) - \pi(A)\Bigr| \to 0,\quad t \to \infty.
$$
\end{defn}

If only $\beta \in (0, 1)$, then there exists a (unique) stationary version of $V(t)$ and therefore of $R_0(t)$. We call the correspoding random variables $(V(\infty), R_0(\infty))$. 

\begin{asmp} We have: $\beta \in (0, 1)$, and 
\begin{equation}
\label{eq:log-negative}
\mathbb E\ln|1 + aV(\infty) + bR_0(\infty)| < 0.
\end{equation}
\label{asmp:stability}
\end{asmp}

The following is the main result of this article.

\begin{theorem}
Consider the following three processes: \begin{itemize}
\item the relative size process $C$ from~\eqref{eq:C-main};
\item the truncated model $(V, R_0, R)$ from from ~\eqref{eq:ln-VIX},~\eqref{eq:n-return},~\eqref{eq:return-main},~\eqref{eq:rel-size};
\item the completed model $(V, R_0, R, P_0, P)$ from ~\eqref{eq:ln-VIX},~\eqref{eq:n-return},~\eqref{eq:n-premia},~\eqref{eq:premia-main},~\eqref{eq:return-main},~\eqref{eq:rel-size}.
\end{itemize}
Under Assumptions~\ref{asmp:innovations},~\ref{asmp:stability}, each of these models is ergodic. 
\label{thm:main} 
\end{theorem}

\begin{proof} {\it Step 1.} We use the main result of \cite{Brandt}: If $\{(A_n, B_n),\, n = 1, 2, \ldots\}$ is stationary, and 
\begin{equation}
\label{eq:Brandt}
\mathbb E\ln|A_n| < 0,\quad \mathbb E\max(\ln|B_n|, 0) < \infty,
\end{equation}
then there exists a stationary version of the process $\{X_n,\, n = 1, 2, \ldots\}$ which satisfies the autoregressive-type equation
\begin{equation}
\label{eq:ar-brandt}
X_{n+1} = A_nX_n + B_n,\, n = 1, 2, \ldots
\end{equation}

 {\it Step 2.} Under Assumption~\ref{asmp:innovations} and $\beta \in (0, 1)$, let us show that the process $V$ has a unique stationary distribution. Indeed, we let $A_n = \beta$ and $B_n = \alpha + W(n)$ in~\eqref{eq:ar-brandt}, then $X_n = \ln V(n)$. Check the conditions~\eqref{eq:Brandt}. It is easy to check that for some $\kappa > 0$ an elementary inequality holds:
\begin{equation}
\label{eq:elementary}
\max(\ln|x|, 0) \le \kappa x^2,\quad x \in \mathbb R.
\end{equation}
From $\mathbb E[W^2(n)] < \infty$, it follows that $\mathbb E[(W(n) + \alpha)^2] < \infty$. Therefore, applying~\eqref{eq:elementary} to the latter inequality, we get: 
$\mathbb E[\max(\ln|\alpha + W(n)|, 0)] < \infty$. We conclude that $\ln V$ has a stationary distribution, and therefore $V$ does as well. Indeed, the function $\ln : (0, \infty) \to \mathbb R$ is one-on-one and continuous in both sides. Such functions preserve stationary distributions and long-term convergence in total variation (which implies ergodicity) for Markov processes. 

{\it Step 3.} Under~\eqref{eq:n-return} and~\eqref{eq:n-premia}, $(P_0, R_0, V)$ has a unique stationary distribution. The second condition in Assumption~\ref{asmp:stability} is taken for this stationary distribution.

{\it Step 4.} Let us show that $C$ from~\eqref{eq:C-main} is stationary. Apply the main result of \cite{Brandt} again and verify~\eqref{eq:Brandt}. In the notation of \cite{Brandt}, we have $A_n := 1 + aV(n) + bR_0(n)$ and $B_n := V(n)(m + \varepsilon(n))$. Therefore, Assumption~\ref{asmp:stability} ensures that $\mathbb E\ln|A_n| < 0$. We need to show that
\begin{equation}
\label{eq:Brandt-new}
\mathbb E\max(\ln|B_n|, 0) = \mathbb E\max(\ln|V(t)(m+\varepsilon(t))|, 0) < \infty.
\end{equation}
For any real number $c \ne 0$, we have:
\begin{equation}
\label{eq:ineq-0}
\ln|c| \le |c|.
\end{equation}
Applying~\eqref{eq:ineq-0} to the left-hand side of~\eqref{eq:Brandt-new}. 
\begin{equation}
\label{eq:logs}
\ln|V(t)(m+\varepsilon(t))| = \ln|m + \varepsilon(t)| + \ln V(t) \le |m+\varepsilon(t)| + \ln V(t).
\end{equation}
Next, for any real numbers $c_1, c_2$,  
\begin{equation}
\label{eq:ineq-1}
\max(c_1 + c_2, 0) \le \max(c_1, 0) + \max(c_2, 0) \le \max(c_1, 0) + |c_2|.
\end{equation}
From~\eqref{eq:ineq-1} and~\eqref{eq:logs}, the left-hand side of~\eqref{eq:Brandt-new} is dominated by
\begin{equation}
\label{eq:LHS}
\mathbb E\max(|m + \varepsilon(t)|, 0) + \mathbb E|\ln V(t)|.
\end{equation}
The innovations $\varepsilon$ have finite second moment by Assumption~\ref{asmp:innovations}. Therefore, 
\begin{equation}
\label{eq:eps}
\mathbb E\max(|m + \varepsilon(t)|, 0) < \infty.
\end{equation}
Next, $\ln V$ is governed by an autoregression of order 1 with innovations $W$ having finite second moment. Thus the stationary distribution of $\ln V$ also has finite second moment and therefore finite first moment. Together with~\eqref{eq:eps}, this proves that~\eqref{eq:LHS} is finite. This proves~\eqref{eq:Brandt-new} and with this the stationarity of $C$. 

{\it Step 5.} Further, $(V, R_0, P_0, C)$ is stationary. Together with~\eqref{eq:premia-main} and~\eqref{eq:return-main}, this proves stationarity of the truncated and the completed models:
\begin{equation}
\label{eq:ln-MC}
(\ln V, R_0, R)\quad \mbox{and}\quad (\ln V, R_0, R, P_0, P).
\end{equation}

{\it Step 6.} Finally, let us show ergodicity for the completed model. The transition function $\mathcal Q$ of this five-dimensional Markov chain is strictly positive: For any set $E \subseteq \mathbb R^5$ of positive Lebesgue measure, the transition probability $\mathcal Q(\mathbf{x}, E) > 0$ for every $\mathbf{x} \in \mathbb R^5$. Indeed, this transition probability is a push-forward of the distribution of the innovations in $\mathbb R^d$ under a certain smooth bijection $\mathbb R^5 \to \mathbb R^5$. And for every $\mathbf{x} \in \mathbb R^5$, this probability measure $\mathcal Q(\mathbf{x}, \cdot)$ and the Lebesgue measure on $\mathbb R^5$ are mutually absolutely continuous. For the rest of this proof, we refer the reader to the classic book \cite{MT-book} for terminology. It is straightforward to show that this Markov chain is irreducible and aperiodic. We have already shown it has a stationary distribution. Therefore, this Markov chain is positive Harris recurrent, and by \cite[Theorem 13.0.1]{MT-book}, this Markov chain is ergodic. 

{\it Step 7.} Exactly the same argument works for the truncated three-dimensional version, or for the relative size process in one dimension. Since we proved these stability and ergodicity results for~\eqref{eq:ln-MC}, they are all true if we replace $\ln V$ with $V$ in~\eqref{eq:ln-MC}. Indeed, the function $\ln : (0, \infty) \to \mathbb R$ is one-on-one and continuous in both sides. Such functions preserve stationary distributions and long-term convergence in total variation (which implies ergodicity) for Markov processes. 
\end{proof}

\section{Stochastic Portfolio Theory}

\subsection{Capital distribution curve} Consider the benchmark and $N$ portfolios $1, \ldots, N$. We model each pair (benchmark, portfolio $k$) with this model. The model must be the same for all $k$. We let $S_k(t)$ be the market capitalization (size) for the $k$th portfolio, and $S_0(t)$ the market size of the benchmark. We define {\it market weights} as 
\begin{equation}
\label{eq:market-weights}
\mu_{k}(t) = \frac{S_k(t)}{S_0(t) + S_1(t) + \ldots + S_N(t)},\quad k = 0, 1, \ldots, N.
\end{equation}
We can also rank them from top to bottom:
\begin{equation}
\label{eq:ranked-market-weights}
\mu_{(0)}(t) \ge \ldots \ge \mu_{(N)}(t).
\end{equation}
Market weights and portfolios based on them are the main topic of Stochastic Portfolio Theory in both discrete time, see \cite{Wong2, Wong1} and continuous time, see the classic book \cite{FernholzBook} and a more recent but also classic survey \cite{FKSurvey}. 

Define by $\delta_k$ and $\varepsilon_k$ the sequences of innovations for equity premium, as in~\eqref{eq:premia-main}, and price returns, as in~\eqref{eq:return-main}, for the $k$th portfolio, $k = 1, \ldots, N$:
\begin{align}
\label{eq:many}
\begin{split}
\frac{P_k(t)}{V(t)} &= M + AC_k(t) + (1 + BC_k(t))\frac{P_0(t)}{V(t)} + \delta_k(t);\\
\frac{R_k(t)}{V(t)} &= m + aC_k(t) + (1 + bC_k(t))\frac{R_0(t)}{V(t)} + \varepsilon_k(t);\\
R_k(t) &= \ln\frac{S_k(t)}{S_k(t-1)},\quad C_k(t) = \ln\frac{S_k(t)}{S_0(t)}.
\end{split}
\end{align}
Together,~\eqref{eq:ln-VIX},~\eqref{eq:n-return},~\eqref{eq:n-premia},~\eqref{eq:many} for a time series Markov model for $2N+3$ series, including the market volatility; one benchmark; and $N$ target portfolios. 

\begin{asmp}
The $2N+3$-dimensional vectors 
$$
(W(t), Z(t), U(t), \delta_1(t), \ldots, \delta_N(t), \varepsilon_1(t), \ldots, \varepsilon_N(t))
$$
are independent identically distributed with mean zero, finite second moment, with strictly positive Lebesgue density on $\mathbb R^{2N+3}$.  
\label{asmp:innov}
\end{asmp}

\begin{theorem} For the model ,~\eqref{eq:ln-VIX},~\eqref{eq:n-return},~\eqref{eq:n-premia},~\eqref{eq:many}, under Assumptions~\ref{asmp:stability},~\ref{asmp:innov}, the process of market weights $(\mu_{0}, \ldots, \mu_{N})$ from~\eqref{eq:market-weights} is ergodic. 
\label{thm:weights}
\end{theorem}

This is the main stability result. It has the meaning that if we have several portfolios, they stay in the long run as one {\it cloud}, and do not split into several {\it clouds}.

\begin{proof}
The relative size processes are time-homogeneous Markov and ergodic:
\begin{equation}
\label{eq:rel-sizes}
C_k(t) = \ln\frac{S_k(t)}{S_0(t)},\quad k = 1, \ldots, N.
\end{equation}
Their vector $(C_1, \ldots, C_N)$ is also time-homogeneous Markov and ergodic: Follows from Assumption~\ref{asmp:innov} in the same way as in Theorem~\ref{thm:main}. And there exists a one-to-one continuous mapping $(C_1, \ldots, C_N) \mapsto (\mu_0, \mu_1, \ldots, \mu_N)$, between $\mathbb R^N$ and the $N$-dimensional simplex 
$$
\{(m_0, \ldots, m_N) \in [0, \infty)^{N+1}\mid m_0 + \ldots + m_N = 1\}.
$$
Thus the process of market weights from~\eqref{eq:market-weights} is ergodic. 
\end{proof}

Thus the ranked market weights process also has a stationary distribution and converges to this distribution in the long run, regardless of the initial conditions. In this stationary distribution, we can plot these ranked market weights versus their ranks on the log scale:
$$
\left((\ln(n+1), \ln\mu_{(n)}(t)),\quad n = 0, \ldots, N\right)
$$
This plot is called the {\it capital distribution curve}. With real-world markets, this curve is linear on most span, and concave overall. Moreover, it shows remarkable long-term stability. See the famous picture in \cite[Chapter 4]{FernholzBook} for eight capital distribution curves at end of years 1929, 1939, \ldots, 1999; see the same picture as  \cite[Figure 13.4]{FKSurvey}. 

In our previous article \cite{CAPM-size}, we captured the observation that well-diversified portfolios of small stocks have higher risk but higher return than that of large stocks. We reproduced this shape of capital distribution curve in \cite{CAPM-size} using simulation. 

From~\eqref{eq:rel-sizes} and~\eqref{eq:market-weights}, we get a simple expression of $\mu_k(t)$ from $C_k(t)$:
$$
\ln\mu_k(t) = C_k(t) + \ln S_0(t) - \ln(S_0(t)+\ldots + S_N(t)),\, k  = 1, \ldots, N.
$$
Therefore, ranking market weights $\mu_k(t), k = 1, \ldots, N$ from top to bottom at any fixed time $t$ is equivalent to ranking relative size terms $C_k(t), k = 1, \ldots, N$ from top to bottom:  $C_{(1)}(t) \ge \ldots \ge C_{(N)}(t)$. Thus instead of plotting the (slightly modified) capital distribution curve $(\ln k, \ln\mu_{(k)}(t)), k = 1, \ldots, N$, we can plot $(\ln k, C_{(k)}(t)), k = 1, \ldots, N$. 

The linearity of the capital distribution curve was rigorously proved in  \cite{Chatterjee} for competing Brownian particles and disproved in \cite{Pal} for volatility-stabilized models. These two types of continuous-time models both capture the property that small stocks have higher risk but higher return than large stocks. But these models do not use CAPM.

\subsection{Reduction to normal order statistics} Fix a constant $k = 1, 2, \ldots$. Assume the market weights $\mu_n$, or, equivalently, relative size $C_n$ are in the stationary distribution, which (by \textsc{Theorem}~\ref{thm:weights}) is limiting distribution. To stress this, we write $t = \infty$ for time argument. Thus we have the ranked (sorted, ordered from top to bottom) values of the relative size process in the stationary distribution:
$$
C_{(1)}(\infty) > \ldots > C_{(N)}(\infty).
$$
Here, $N$ is the overall number of portfolios (excluding the benchmark). We are interested in the joint distribution of these sorted relative size values. Assume $\mathcal Z_1, \ldots, \mathcal Z_N \sim \mathcal N(0, 1)$ IID is the standard normal sample. 

\begin{theorem} If we assume that $\varepsilon_k$ are Gaussian and independent identically distributed, there exist random variables $\mathcal M$ and $\mathcal S > 0$ independent of the point process which are functions of two time series: $V$ and $Z$, and
$$
\left[\frac{C_1(\infty) - \mathcal M}{\mathcal S}, \ldots, \frac{C_N(\infty) - \mathcal M}{\mathcal S}\right] \stackrel{d}{=} (\mathcal Z_1, \ldots, \mathcal Z_N).
$$
\label{thm:normal-curve}
\end{theorem}

\begin{figure}[t]
\includegraphics[width = 14cm]{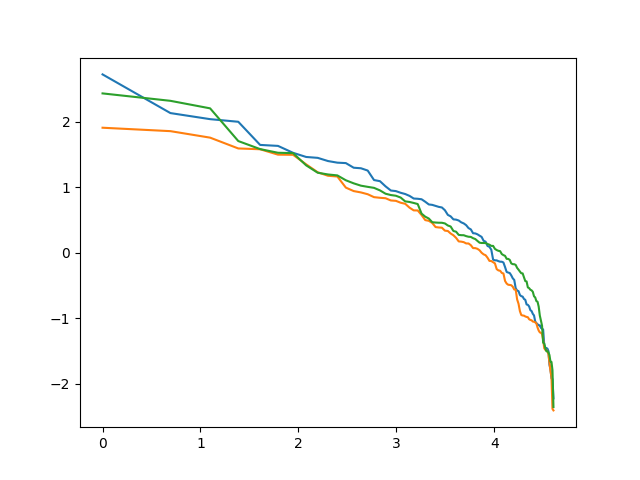}
\caption{Simulations of the standard normal market curve $(\ln k, \mathcal Z_{(k)})$ for $k = 1,\ldots, N$ if $\mathcal Z_1, \ldots, \mathcal Z_N \sim \mathcal N(0, 1)$ IID sample, for $N = 100$.}
\label{fig:normal-curve}
\end{figure}

\begin{proof}
Fix a $k = 1, \ldots, N$. Apply \cite[(0.6)]{Brandt} to express the stationary distribution of the relative size process $C_k$ from~\eqref{eq:C-main}: letting 
\begin{align}
\label{eq:A-B}
\begin{split}
A_n &:= 1 + aV(n) + bR_0(n) = 1 + aV(n) + bV(n)(g+U(n)),\\ 
B_n &:= V(n)(m + \varepsilon_k(n)),
\end{split}
\end{align}
pick any $n = 1, 2, \ldots$ and get:
\begin{equation}
\label{eq:C-A-B}
C_k(\infty) := B_n + A_nB_{n-1} + A_nA_{n-1}B_{n-2} + \ldots 
\end{equation}
We assume $V(n)$ and $U(n)$ are defined for all $n \in \mathbb Z$, not just $n \ge 0$. Assume $\sigma_{\varepsilon}$ is the standard deviation of innovations $\varepsilon_k$. From~\eqref{eq:A-B} and~\eqref{eq:C-A-B}, given time series $V$ and $Z$, the stationary distribution of $C_k(\infty)$ is Gaussian with mean and variance 
\begin{align*}
\mathcal M &:= m\left[V(n) + A_nV(n-1) + A_nA_{n-1}V(n-2) + \ldots\right],\\
\mathcal S^2 &:= \sigma^2_{\varepsilon}\left[V^2(n)  + A_n^2V^2(n-1) + A_n^2A_{n-1}^2V^2(n-2) + \ldots\right].
\end{align*}
Given time series $V$ and $Z$, the random variables $C_1(\infty), \ldots, C_N(\infty)$ are independent. Thus their standardized versions are (conditionally on $V$ and $Z$) IID standard Gaussian:
\begin{equation}
\label{eq:standardize}
\mathcal Z_n := \frac{C_n(\infty) - \mathcal M}{\mathcal S},
\quad n = 1, \ldots, N.
\end{equation}
\end{proof}

This motivates the following definition. 

\begin{defn} We define the {\it standard normal market curve} 
\begin{equation}
\label{eq:standard-normal-curve}
(\ln k, \mathcal Z_{(k)}),\quad k = 1,\ldots, N,
\end{equation}
for standard normal sample $\mathcal Z_1, \ldots, \mathcal Z_N$ and its ranked version $\mathcal Z_{(1)} \ge \ldots \ge \mathcal Z_{(N)}$.
\end{defn}

The significance of Theorem~\ref{thm:normal-curve} is as follows. Assume we rank $\mathcal Z_{(1)} > \ldots > \mathcal Z_{(N)}$. An immediate consequence of Theorem~\ref{thm:normal-curve} is that the (slightly modified) capital distribution curve $(\ln k, C_{(k)})$ for $k = 1, \ldots, N$ has almost the same shape as the standard normal market curve. The only difference is a shift and change in slope, both random but independent of the standard normal market slope itself. Multiplication by $\mathcal S$ preserves ordering, since $\mathcal S > 0$. We plotted three simulations of such curve in \textsc{Figure}~\ref{fig:normal-curve}. The Python code for this simulation in \textsc{Figure}~\ref{fig:normal-curve} is available on \texttt{GitHub} repository \texttt{asarantsev/size-capm-vix}, file \texttt{normal-sample-simulation.py} which reproduces the real-world shape of the curve. The changes in random constants $\mathcal M$ and $\mathcal S$ make this curve shift and change its slope. But they do not influence its overall shape. 

\subsection{Capital distribution curve from the standard normal distribution} It remains to study the behavior of the curve~\eqref{eq:standard-normal-curve} using the classic Extreme Value Theory. Most of the definitions and results of this subsection are well-known. We do not even attempt to provide an exhaustive list of references. Instead, we mention classic monograph \cite{Resnick-Book} and a classic textbook \cite{Extreme-Book}. For Poisson point processes on the real line, see another classic monograph \cite{Kingman}. Pick a function $\lambda : \mathbb R \to [0, \infty)$ which satisfies 
$$
\Lambda(u) := \int_{u}^{\infty}\lambda(x)\,\mathrm{d}x < \infty,\quad u \in \mathbb R;\quad \Lambda(-\infty) = \infty.
$$

\begin{defn}  A {\it Poisson point process} on $\mathbb R$ with {\it intensity} or {\it rate} $\lambda$ is defined as a decreasing sequence of random numbers $\mathcal N^+_1 > \mathcal N^+_2 > \ldots$ such that the (random) number of points on any interval $[a, b]$ has the Poisson distribution with mean $\int_a^b\lambda(u)\,\mathrm{d}u$. 
\end{defn}

For example, take $\lambda(t) = e^{-t}$. Denote by $\tau_k$ the $k$th jump time of the standard Poisson process on $[0, \infty)$: $\tau_k - \tau_{k-1}$ are IID exponential with mean 1, where by convention $\tau_0 := 0$. For fixed $m$, we can also express 
(see  \cite[Chapter 8, Exercises 6 and 7]{Extreme-Book}):
\begin{equation}
\label{eq:rescaling+}
\mathcal N^+_k = -\ln(\tau_k),\quad k = 1, \ldots, m.
\end{equation}

\begin{defn} The (standard) {\it Gumbel distribution} $G$ is defined by its cumulative distribution function $\exp(-\exp(-x))$.
\end{defn}

It is known, see classic references \cite[Theorem 8.3.1, Example 8.3.4]{Extreme-Book} or \cite[Chapter 1]{Resnick-Book}, that the normal distribution belongs to the {\it Gumbel domain of attraction:} Consider the maximum of $n$ IID normal variables:
$$
M_n = \max(X_1, \ldots, X_n),\quad X_1, X_2, \ldots \sim \mathcal N(0, 1)\ \mbox{IID.}
$$
After scaling, this maximum converges weakly to Gumbel distribution as $n \to \infty$:
\begin{equation}
\label{eq:Gumbel}
\frac{M_n - b_n}{a_n} \xrightarrow{d} G,\quad n \to \infty,
\end{equation}
where $a_n > 0$ and $b_n$ are suitable constants. A common suggestion is 
$$
a_n = \frac{1}{\sqrt{2\ln n}},\quad b_n = \sqrt{2\ln n} - \frac{\ln(4\pi\ln n)}{2\sqrt{2\ln n}}.
$$
But the convergence rate for this choice of constants is not the best, as shown in \cite[Example 8.3.4]{Extreme-Book}. There are ways to improve this, for example \cite{Hall}:
$$
n\varphi(b_n) = b_n,\quad a_n = 1/b_n,\quad \varphi(u) := \frac1{\sqrt{2\pi}}\exp\left[-\frac{u^2}2\right].
$$
For any such sequences $(a_n)$ and $(b_n)$, we have convergence of top $k$ ranked standardized variables $X_{(1)} > \ldots > X_{(k)}$ to the rightmost $k$ points of $\mathcal N$, as $N \to \infty$:
\begin{equation}
\label{eq:multiple-conv+}
\left[\frac{X_{(1)} - b_N}{a_N}, \ldots, \frac{X_{(k)} - b_N}{a_N}\right] \xrightarrow{d} \left[\mathcal N^+_1, \ldots, \mathcal N^+_k\right].
\end{equation}
Similarly and symmetrically, as $N \to \infty$, 
\begin{equation}
\label{eq:multiple-conv-}
\left[\frac{X_{(N)} + b_N}{a_N}, \ldots, \frac{X_{(N-k+1)} + b_N}{a_N}\right] \xrightarrow{d} -\left[\mathcal N^+_1, \ldots, \mathcal N^+_k\right].
\end{equation}
This convergence~\eqref{eq:multiple-conv+} or~\eqref{eq:multiple-conv-} follows from \cite[Theorem 8.4.2]{Extreme-Book}, or \cite{Miller}, or \cite[Chapter 4]{Resnick-Book}; see also \cite[Chapter 8, Exercises 6--8]{Extreme-Book} and compare with~\eqref{eq:rescaling+}. Recall~\eqref{eq:rescaling+} and plot this Poisson point process with $x$-axis log scale. This represents (up to shift by $b_N$ and rescaling by $a_N$) the left upper end of the capital distribution curve:
\begin{equation}
\label{eq:standard-curve+}
(\ln k, \mathcal N^+_k) = (\ln k, -\ln(\tau_k)),\, k = 1, 2, \ldots, m.
\end{equation}

\begin{figure}[t]
\subfloat[Upper Left End]{\includegraphics[width = 8cm]{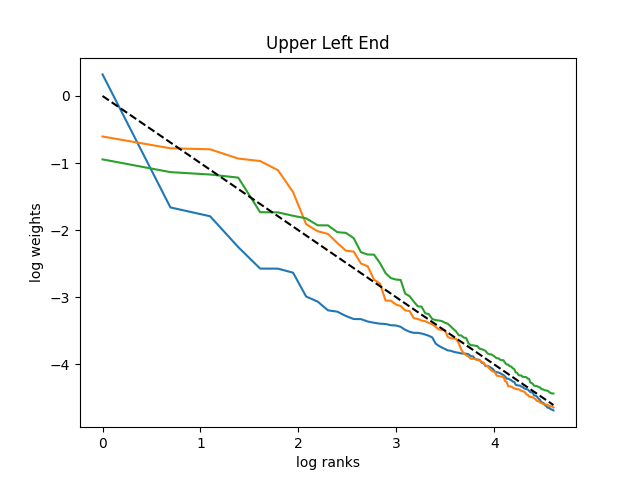}}
\subfloat[Lower Right End]{\includegraphics[width = 8cm]{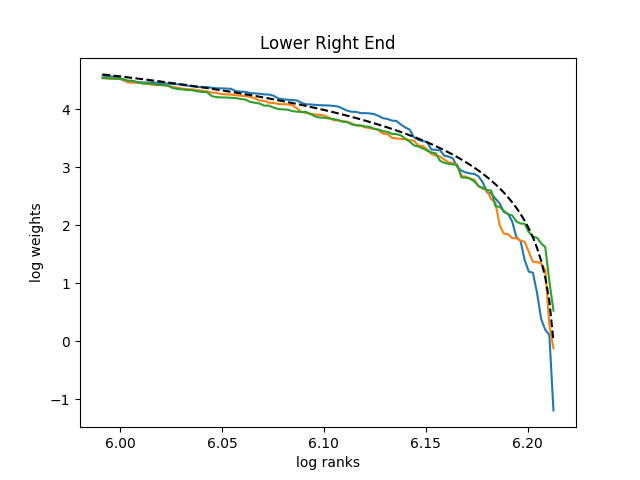}}
\caption{Simulated capital distribution curves. Left panel: $(-\ln k, -\ln\tau_k)$ for $k = 1, \ldots, 100$ and $y = -x$ for $x \in [0, \ln100]$. Right panel: $(\ln(501-k), -\ln\tau_k)$ for $k = 1, \ldots, 100$ and $y = \ln(500 - e^x)$ for $x \in [\ln 401, \ln 500]$. For each panel, 3 simulations (solid) and the deterministic function (dotted).}
\label{fig:ends-sim}
\end{figure}

The next lemma is the key to our analysis of both ends of the standard normal curve~\eqref{eq:standard-normal-curve}. 

\begin{lemma} 
With probability $1$, the sequence $|\ln\tau_k - \ln k|,\, k = 1, 2, \ldots$ is bounded.
\label{lemma:diff-logs}
\end{lemma}

\begin{proof} From definitions of $\tau_k$, we have: $\tau_k$ and $\tau_{k+1} - \tau_k$ are independent; $\tau_k$ has Gamma distribution (sum of $k$ IID exponential random variables with mean $1$); $\tau_{k+1} - \tau_k$ is another exponential random variable with mean $1$. By the Strong Law of Large Numbers, $\tau_k/k \to 1$ almost surely. Taking the logarithms, we get $\ln \tau_k - \ln k \to 0$ almost surely. Every convergent sequence is bounded. This completes the proof. 
\end{proof}

Thus we can replace $\ln\tau_k$ in~\eqref{eq:standard-curve+} with $\ln k$. This result allows us to approximate the curve~\eqref{eq:standard-curve+} with continuous functions. The upper-left end of the curve~\eqref{eq:standard-curve+} then becomes $(\ln k, -\ln k)$. This is a linear curve with 45 degrees of incline. Next, the lower-right end of the curve becomes $(\ln k, \ln(N+1-k))$. Let $x = \ln k$ and $y = \ln(N+1-k)$: Then
$$
y = \ln(a - e^x),\quad a := N+1.
$$
This function is concave but obviously not linear: 
$$
y'' = \left(\frac{(a-e^x)'}{a - e^x}\right)' = \left(1 + \frac{a}{e^x-a}\right)' = -\frac{ae^x}{(e^x-a)^2} < 0.
$$
In light of all this, we reproduce the shape of this capital distribution curve. See \textsc{Figure}~\ref{fig:ends-sim} for the upper left and lower right ends of the capital distribution curve. The simulation of $\ln\tau_k$ for $k = 1, \ldots, 100$ and the plot of the comparative deterministic function (with $N = 500$) was done in Python. The code \texttt{standard-curve-simulation.py} is in \texttt{GitHub} repository \texttt{asarantsev/size-capm-vix} The shapes of both ends of the curve in \textsc{Figure}~\ref{fig:normal-curve} are reproduced here in \textsc{Figure}~\ref{fig:ends-sim}.

\section{Conclusion and Further Research} 

We combined main results of our previous articles \cite{CAPM-size, VIX}. In \cite{CAPM-size}, the main model (CAPM plus linear dependence of $\alpha$ and $\beta$ upon relative size) was used to capture the property that small stocks have, on average, higher risks and returns. In this article, we add stochastic volatility to this model. A few, but not all, portfolios have the following property: Their linear regression equations (combining CAPM and size) have residuals closer to IID Gaussian if we include VIX as the multiplicative term in the equation. We state and prove a simple sufficient condition (\textsc{Theorem}~\ref{thm:main}) for ergodicity. We make some connections with SPT, adding our model to the collection of proposed models capturing this size effect. We state and prove rigorous results on the capital distribution curve. We fill two lacunas in our previous research from \cite{CAPM-size}. 

{\it First lacuna:} Long-term stability results apply to the case of constant volatility from \cite{CAPM-size}, but not to the case of stochastic volatility studied here. 

{\it Second lacuna:} The curve is the order statistics of the normal sample (but with random mean and standard deviation). Such curve reproduces the real-world shape of the capital distribution curve: linear at the upper left end, and concave at the lower right end. We investigated this only here and not in \cite{CAPM-size}.

For future research, we could fit non-normal distributions for innovation series $W, \delta, \varepsilon$, and check conditions of \textsc{Theorem}~\ref{thm:main} for the resulting distributions. We could also study the rate of convergence to the stationary distribution as $t \to \infty$. 

Also, we could include the {\it value effect:} Stocks priced cheaply to fundamentals (earnings, dividends, book value) tend to outperform other stocks, on average. We would include, for example, dividend yield as a factor in $\alpha$ and $\beta$ from CAPM. To make the model complete, we need to model dividend yield separately (for example, as an autoregression). %Our goal is to statistically fit this model and prove long-term stability. 

Finally, we could rewrite this model for continuous time (which is straightforward, since our models are Markov) and study the new model.

\end{document}